\documentclass{llncs}

\usepackage[ruled,vlined,linesnumbered,commentsnumbered]{algorithm2e}
\usepackage{algcompatible}
\usepackage{amsmath,amssymb,amsfonts}

\usepackage{graphicx}

\newcommand {\ignore} [1] {}

\def\AEC     {{\sc Activation Edge-Cover}}
\def\FL        {{\sc Facility Location}}
\def\MPEC  {{\sc Min-Power Edge-Cover}}
\def\GMC   {{\sc Generalized Min-Covering}}
\def\SC       {{\sc Set-Cover}}

\def\opt  {\mathsf{opt}}

\def\di   {\displaystyle}

\def\subs   {\subseteq}
\def\sem   {\setminus}
\def\empt {\emptyset}

\def\f {\frac}

\def\al {\alpha}
\def\th {\theta}
\def\om {\omega}
\def\bom {\bar{\omega}}
\def\si {\sigma}
\def\eps {\epsilon}

\def\De {\Delta}

\def\a       {{\bf a}}
\def\b      {{\bf b}}
\def\c      {{\bf c}}
\def\d       {{\bf d}}
\def\q       {{\bf q}}
\def\t       {{\bf t}}
\def\w       {{\bf w}}

\begin{document}

\title{Approximating activation edge-cover and \\ facility location problems}

\author{Zeev Nutov \and Eli Shalom} 

\institute{The Open University of Israel. \email{nutov@openu.ac.il, eli.shalom@gmail.com}}

\maketitle

\begin{abstract}
What approximation ratio can we achieve for the {\FL} problem if whenever a client $u$ connects to a facility $v$,
the opening cost of $v$ is at most $\th$ times the service cost of $u$? 
We show that this and many other problems are a particular case of the {\AEC} problem.
Here we are given a multigraph $G=(V,E)$, a set $R \subs V$ of terminals, 
and thresholds $\{t^e_u,t^e_v\}$ for each $uv$-edge $e \in E$.
The goal is to find an assignment $\a=\{a_v:v \in V\}$ to the nodes minimizing $\sum_{v \in V} a_v$,
such that the edge set $E_\a=\{e=uv: a_u \geq t^e_u, a_v \geq t^e_v\}$ activated by $\a$ covers~$R$.
We obtain ratio $1+\om(\th) \approx \ln \th-\ln \ln \th$ for the problem, 
where $\om(\th)$ is the root of the equation $x+1=\ln(\th/x)$ and $\th$ is a problem parameter.
This result is based on a simple generic algorithm for the problem of minimizing 
a sum of a decreasing and a sub-additive set functions, which is of independent interest.
As an application, we get that the above variant of {\FL} admits ratio $1+\om(\th)$;
if for each facility all service costs are identical then we show a better ratio 
$\di 1+\max_{k \geq 1} \f{H_k-1}{1+k/\th}$, where $H_k=\sum_{i=1}^k 1/i$. 
For the {\MPEC} problem we improve the ratio $1.406$ of \cite{CKN} 
(achieved by iterative randomized rounding) to $1+\om(1)<1.2785$.
For unit thresholds we improve the ratio $73/60 \approx 1.217$ of \cite{CKN} to $\f{1555}{1347} \approx 1.155$.
\end{abstract}


\section{Introduction} \label{s:intro}

Let $G=(V,E)$ be an undirected multigraph where each edge $e \in E$ has an {\bf activating function}
$f^e$ from some range $L^e \subs \mathbb{R}^2_+$ to $\{0,1\}$.
Given a non-negative {\bf assignment} $\a=\{a_v:v \in V\}$ to the nodes, we say that 
a $uv$-edge $e \in E$ is {\bf activated} by $\a$ if $f^e(a_u,a_v)=1$.
Let $E_\a=\{e \in E:f^e(a_u,a_v)=1\}$ denote the set of edges activated by $\a$.
The {\bf value} of an assignment $\a$ is $\a(V)=\sum_{v \in V}a_v$. 
In {\sc Activation Network Design} problems the goal is to find an assignment $\a$ of minimum value,
such that the edge set $E_\a$ activated by $\a$ satisfies a prescribed property.
We refer the reader to a paper of Panigrahi \cite{P} and a recent survey \cite{N-act} on activation problems,
where also the following two assumptions are justified.

\medskip

\noindent
{\bf Monotonicity Assumption.}
{\em For every $e \in E$, $f^e$ is monotone non-dec\-reasing, namely,
$f^e(x_u,x_v)=1$ implies $f^e(y_u,y_v)=1$ if $y_u \geq x_u$ and $y_v \geq x_v$.}

\medskip

\noindent
{\bf Polynomial Domain Assumption.} 
{\em Every $v \in V$ has a polynomial size in $n=|V|$ set $L_v$ of ``levels''
and $L^e=L_u \times L_v$ for every $uv$-edge $e \in E$.}

\medskip

Given a set $R \subs V$ of {\bf terminals} we say that an edge set $J$ is an {\bf $R$-cover}
or that {\bf $J$ covers $R$} if every $v \in R$ has some edge in $J$ incident to it.
In the {\sc Edge-Cover} problem we seek an $R$-cover of minimum value. 
The min-cost version of this problem can be solved in polynomial time \cite{Ed},
and it is one of the most fundamental problems in Combinatorial Optimization, cf. \cite{Sch}.

We consider the {\AEC} problem. 
Since we consider multigraphs, $e=uv$ means that $e$ is a $uv$-edge,
namely, that $u,v$ are the endnodes of $e$; $e=uv \in E$ means that $e \in E$ is a $uv$-edge.
Under the two assumptions above, the problem can be can be formulated without activating functions.
For this, replace each edge $e=uv$ by a set of at most $|L_u| \cdot |L_v|$
$uv$-edges $\{e(t_u,t_v): (t_u,t_v) \in L_u \times L_v, f^e(t_u,t_v)=1\}$.
Then for any $J \subs E$ the optimal assignment $\a$ activating $J$ is given by 
$a_u=\max\{t^e_u: e \in J \mbox{ is incident to } u\}$;
here and everywhere a maximum or a minimum taken over an empty set is assumed to be zero.
Consequently, the problem can be restated as follows. 

\begin{center} \fbox{\begin{minipage}{0.97\textwidth} \noindent
\underline{\AEC} \\
Input: \ A graph $G=(V,E)$, a set of terminals $R \subs V$, and  thresholds $\{t^e_u,t^e_v\}$ for each $uv$-edge $e \in E$. \\
Output: An assignment $\a$ of minimum value $\a(V)=\sum_{v \in V}a_v$,
such that the edge set $E_\a=\{e=uv \in E: a_u \geq t^e_u,a_v \geq t^e_v\}$ activated by $\a$ covers $R$.
\end{minipage}}\end{center}

As we will explain later, {\AEC} problems are among the most fundamental problems in network design,
that include NP-hard problems such as {\sc Set-Cover}, {\FL},
covering problems that arise in wireless networks (node weighted/min-power/installation problems), and many other problems.

To state our main result we define assignments $\q$ and $\c$, where $c_v=q_v=0$ if $v \in V \sem R$ and for $u \in R$:
\begin{itemize}
\item[$\bullet$]
$\di q_u=\min_{e =uv \in E} t^e_u$ is the minimum threshold at $u$ of an edge in $E$ incident to $u$.
\item[$\bullet$]
$\di c_u=\min_{e =uv \in E}(t^e_u+t^e_v)-q_u$, so $c_u+q_u$ is the minimum value of an edge in $E$ incident to $u$.
\end{itemize}
The quantity $\max_{u \in R} c_u/q_u$ is called the {\bf slope} of the instance.
We say that an {\AEC} instance is {\bf $\th$-bounded} if the instance slope is at most $\th$, namely if $c_u \leq \th q_u$ for all $u \in R$;
moreover, we assume by default that $\th=\max_{u \in R} c_u/q_u$ is the instance slope.
For each $u \in R$ let $e_u$ be some minimum value edge covering $u$.
Then $\{e_u: u \in R\}$ is an $R$-cover of value at most $(\c+\q)(R)$.
From this and the definition of $\th$ we get 
$$
0 \leq {\opt}-\q(R) \leq \c(R) \leq \th \q(R) \leq \th {\opt} 
$$
In particular, $(\c+\q)(R) \leq (\th+1){\opt}$. Using this, it is possible to design a greedy algorithm with ratio $1+\ln(\th+1)$. 
We will show how to obtain a better ratio (the difference is quite significant when $\th \leq 10^4$ -- see Table~\ref{tbl:th}), as follows.

\begin{table} [htbp]
\begin{center}
\begin{tabular}{|c|c|c|c|c|c|c|c|c|c|c|c|c|c|}  \hline  
$\th$                         & 1          & 2          & 3          & 4          & 5          & 10        & 100      & 1000    & 10000 & 1000000
\\\hline
$1+\om(\th)$           & 1.2785 & 1.4631 & 1.6036 & 1.7179 & 1.8146 & 2.1569 & 3.6360 & 5.4214 & 7.3603 & 11.4673
\\\hline
$1+\bar{\om}(\th)$ & 1.2167 & 1.3667 & 1.4834 & 1.5800 & 1.6637 & 1.9645 & 3.3428 & 5.0808 & 6.9967 & 11.0820
\\\hline
$\ln \th-\ln\ln \th$ & -           & 1.0597 & 1.0046 & 1.0597 & 1.1336 & 1.4686 & 3.0780 & 4.9752 &  6.9901 & 11.1898
\\\hline
$1+\ln(\th+1)$         & 1.6932 & 2.0987 & 2.3863 & 2.6095 & 2.7918 & 3.3979 & 5.6152 & 7.9088 & 10.2105 & 14.8156
\\\hline
\end{tabular}
\end{center}
\caption{Some numerical bounds on $1+\om(\th)$, $1+\bar{\om}(\th)$, $\ln \th-\ln\ln\th$, and $1+\ln(\th+1)$.}
\label{tbl:th}        \vspace*{-0.5cm}
\end{table}

The {\bf Lambert $W$-Function} (a.k.a. ProductLog Function) $W(z)$ is the inverse function of $f(W)=We^W$.
It is known that for any $\th>0$, $\om(\th)=W(\th/e)$ equals to the (unique) real root of the equation $x+1=\ln(\th/x)$,
and that $\lim_{\th \rightarrow \infty} [1+\om(\th)-(\ln\th-\ln\ln\th)]=0$.
Our main result is:

\begin{theorem} \label{t:EC}  
{\AEC} admits ratio $1+\om(\th)$ for $\th$-bounded instances.
The problem also admits ratio $1+\ln(\De+1)$, and ratio $1+\ln\De$ if $R$ is an independent set in $G$, 
where $\De$ is the maximum number of terminal neighbors of a node in $G$.
\end{theorem}

This result is based on a generic simple approximation algorithm for the problem of minimizing 
a sum of a decreasing and a sub-additive set functions, which is of independent interest;
it is described in the next section. This result is inspired by the so called ``Loss Contraction Algorithm'' of 
Robins \& Zelikovsky \cite{RZ} for the {\sc Steiner Tree} problem,
and the analysis in \cite{GHNP} of this algorithm. 

Let us say that $v \in V$ is a {\bf steady node} if the thresholds $t^e_v$ 
of the edges $e$ adjacent to $v$ are all equal the same number $w_v$, which we call the {\bf weight of~$v$}. 
Note that we may assume that all non-terminals are steady, by replacing each $v \in V \sem R$ by $L_v$ new nodes;
see the so called ``Levels Reduction'' in \cite{N-act}.
This implies that we may also assume that no two parallel edges are incident to the same non-terminal.
Clearly, we may assume that $R \sem V$ is an independent set in $G$.
Let {\sc Bipartite} {\AEC} be the restriction of {\AEC} to instances when also $R$ is an independent set, 
namely, when $G$ is bipartite with sides $R,V \sem R$.
Note that in this case $G$ is a simple graph and all non-terminals are steady.

We now mention some particular threshold types in {\AEC} problems, 
some known problems arising from these types,  
and some implications of Theorem~\ref{t:EC} for these problems. 

\paragraph{\sc Weighted Set-Cover} \ \\
This is a particular case of {\sc Bipartite} {\AEC} when all nodes are steady 
and nodes in $R$ have weight $0$.
Note that in this case $\th$ is infinite,
and we can only deduce from Theorem~\ref{t:EC} the known ratio $1+\ln \De$. 
Consider a modification of the problem, which we call {\sc $\th$-Bounded Weighted Set-Cover}:
when we pick a set $v \in V \sem R$, we need to pay $w_v/\th$ for each element in $R$ covered by $v$.
Then the corresponding {\AEC} instance is $\th$-bounded.

\paragraph{\FL} \ \\
Here we are given a bipartite graph with sides $R$ (clients) and $V \sem R$ (facilities),
weights (opening costs) $\w=\{w_v:v \in V \sem R\}$, and 
distances (service costs) $\d=\{d_{uv}:u \in R, v \in V \sem R\}$.
We need to choose $S \subs V \sem R$ with $\w(S)+\sum_{u \in R} d(u,S)$ minimal,
where $d(u,S)=\min_{v \in S} d_{uv}$ is the minimal distance from $u$ to $S$.
This is equivalent to {\sc Bipartite} {\AEC}. 
Note however that if for some constant $\th$ we have 
$w_v \leq \th d_{uv}$ for all $uv \in E$ with $u \in R$ and $v \in V \sem R$,
then the corresponding {\sc Bipartite} {\AEC} instance is $\th$-bounded, 
and achieves a low constant ratio even for large values of $\th$.

\paragraph{\sc Installation Edge-Cover} \ \\
Suppose that the installation cost of a wireless network 
is proportional to the total height of the towers for mounting antennas. 
An edge $uv$ is activated if the towers at $u$ and $v$ are tall enough to overcome obstructions 
and establish line of sight between the antennas.  
This is modeled as each pair $u,v \in V$ has a height demand $h^{uv}$ 
and constants $\gamma_{uv},\gamma_{vu}$, such that a $uv$-edge is activated 
by $\a$ if the scaled heights $\gamma_{uv}a_u, \gamma_{vu}a_v$ sum to at least $h^{uv}$. 
In the {\sc Installation Edge-Cover} problem, we need to assign heights to the antennas such
that each terminal can communicate with some other node, while minimizing the total sum of the heights.
The problem is {\sc Set-Cover} hard even for $0,1$ thresholds and bipartite $G$ \cite{P}.
But in a practical scenario, the quotient of the maximum tower height over the minimum tower height is usually bounded by a constant;
say, if possible tower heights are $5,15,20$, then the slope is $\th=4$.

\paragraph{\MPEC} \ \\
This problem is a particular case of {\AEC} when $t^e_u=t^e_v$ for every edge $e=uv \in E$;
note that $\th=1$ in this case (in fact, the case $\th=1$ is much more general). 
The motivation is to assign energy levels to the nodes of a wireless network
while minimizing the total energy consumption, and enabling communication for every terminal.
The {\MPEC} problem is NP-hard even if $R=V$, or if
$R$ is an independent set in the input graph $G$ and unit thresholds \cite{HKMN}.
The problem admits ratio $2$ by a trivial reduction to the min-cost case. 
This was improved to $1.5$ in \cite{KN-cov}, and then to $1.406$ in \cite{CKN}, 
where is also given ratio $73/60$ for the bipartite case and for unit thresholds.

\medskip

From Theorem~\ref{t:EC} and the discussion above we get:

\begin{corollary} \label{c:main}
{\MPEC} admits ratio $1+\om(1) < 1.2785$,
and the $\th$-bounded versions of each one of the problems 
{\sc Weighted Set-Cover}, {\FL}, and {\sc Installation Edge-Cover},
admits ratio $1+\om(\th)$.
\end{corollary}

Let us illustrate this result on the {\FL} problem.
One might expect a constant ratio for any $\th>0$, but our ratio $1+\om(\th)$ is surprisingly low.
Even if $\th=100$ (service costs are at least $1\%$ of opening costs) then we get a small ratio $1+\om(100)<3.636$.
Even for $\th=10^4$ we still get a reasonable ratio $1+\om(10^4)<7.3603$.
All previous results for the problem are usually summarized by just two observations: 
the problem is {\sc Set-Cover} hard (so has a logarithmic approximation threshold by \cite{RS,Feige}),
and that it admits a matching logarithmic ratio $1+\ln |R|$ \cite{Chva}; 
see surveys on {\FL} problems by Vygen \cite{Vy} and Shmoys \cite{Shmoys}.
Due to this, all work focused on the more tractable {\sc Metric} {\FL} problem.
Our Theorem~\ref{t:EC} implies that many practical non-metric {\FL} instances admit a reasonable small constant ratio. 

For the case of ``locally uniform'' thresholds -- 
when for each non-terminal (facility) all thresholds (service costs) are identical, we show a better ratio,
see also Table~\ref{tbl:th}. In what follows, let $H_k$ denote the $k$-th harmonic number.

\begin{theorem} \label{t:lu}
{\sc Bipartite} {\AEC} with locally uniform thres\-holds admits ratio $1+\bom(\th)$, where $\bom(\th)=\di \max_{k \geq 1} \f{H_k-1}{1+k/\th}$.
\end{theorem}

We do not have a convenient formula for $\bar{\om}(\th)$,
but in Section~\ref{s:lu} we observe that the maximum is attained for the smallest integer $k_\th$ such that $H_k \geq 2+\f{\th-1}{k+1}$.
We will show that $\bar{\om}(\th)<\om(\th)$ for all $\th$, and 
that $\lim_{\th \rightarrow \infty} [\om(\th)-\bar{\om}(\th)]=0$,
so both $\om(\th)$ and $\bar{\om}(\th)$ are close to $\ln \th-\ln\ln \th$ for large values of $\th$, although the convergence is very slow;
see also Table~\ref{tbl:th}.

We will also show that $\bar{\om}(1) =73/60$.
Note that our Theorem~\ref{t:EC} ratio $1.2785$ for $\th=1$ significantly improves 
the previous best ratio $1.406$ of \cite{CKN} for {\MPEC} on general graphs 
achieved by iterative randomized rounding; we do not match the ratio $73/60$ of \cite{CKN} 
for the bipartite case, but note that the case $\th=1$ is much more general than the min-power 
case considered in \cite{CKN}.

Theorem~\ref{t:lu} has some applications for the {\SC} problem.
Given a {\SC} instance represented by a bipartite graph, obtain a {\sc Bipartite} {\AEC} instance by assigning 
unit weights to sets (non-terminals) and a threshold $\eps$ to every terminal (element).
These threshold are locally uniform and the obtained instance has slope $\th=1/\eps$.
For this instance the Theorem~\ref{t:lu} algorithm coincides with the standard greedy algorithm,
and computes a {\SC} solution of size 
$\leq (1+\bar{\om}(1/\eps))(\tau+n\eps)-n\eps=\tau+\bar{\om}(1/\eps)(\tau+n\eps)$,
where $\tau$ is the optimal size of a {\SC} solution. 
Substituting $\eps=\f{\tau}{nM}$ and dividing by $\tau$ we get that
for any $M>0$ the greedy algorithm for {\SC} achieves ratio $1+\bom(nM/\tau)(1+1/M)$.
We note that Slavik \cite{SL} proved that the greedy algorithm for {\SC} achieves ratio 
$\ln n-\ln\ln n+\Theta(1)$, while our ratio for locally uniform thresholds is $1+\bar{\om}(\th) = \ln \th-\ln\ln \th+\Theta(1)$;
we will discuss the relation between these two results in the full version.

In addition, we consider unit thresholds, and 
using some ideas from \cite{CKN} improve the previous best ratio $73/60$ of \cite{CKN} as follows.

\begin{theorem} \label{t:unit}
{\AEC} with unit thresholds admits ratio $\f{1555}{1347}$.
\end{theorem}

We note that our main contribution is not technical, although some proofs are non-trivial
(the reader may observe that proofs of many seemingly complicated  results 
were substantially simplified with years, by additional effort).
Our main contribution is giving a unified algorithm for a large class of problems that we identify -- {\sc $\th$-Bounded} {\AEC} problems,
either substantially improving known ratios, or showing that many seemingly {\sc Set-Cover} hard problems
may be tractable in practice.
Let us also point out that our main result is more general than the applications listed in Corollary~\ref{c:main}.
The generalization to $\th$-bounded {\AEC} problems is different from earlier results; 
besides finding a unifying algorithmic idea generalizing and improving previous results, 
we are also able to find tractable special cases in a new direction.

The rest of this paper is organized as follows.
In Section~\ref{s:GMC} we define the {\GMC} problem and analyze a greedy algorithm for it, see Theorem~\ref{t:ga}.
In Section~\ref{s:EC} we use Theorem~\ref{t:ga} to prove Theorem~\ref{t:EC}.
Theorems \ref{t:lu} and \ref{t:unit} are proved using a modified method in Sections \ref{s:lu} and \ref{s:unit}, respectively.

\section{The {\GMC} problem} \label{s:GMC}

A set function $f$ is 
{\bf increasing} if $f(A) \leq f(B)$ whenever $A \subs B$; $f$ is {\bf decreasing} if $-f$ is increasing, and $f$ 
is {\bf sub-additive} if $f(A \cup B) \leq f(A)+f(B)$ for any subsets $A,B$ of the ground-set.
Let us consider the following algorithmic problem:

\medskip \begin{center} \fbox{\begin{minipage}{0.96\textwidth} \noindent
\underline{{\GMC}} \\
{\em Input:} \ 
Non-negative set functions $\nu,\tau$ on subsets of a ground-set $U$ such that 
$\nu$ is decreasing, $\tau$ is sub-additive, and $\tau(\empt)=0$. \\
{\em Output:}
$A \subs U$ such that $\nu(A)+\tau(A)$ is minimal.
\end{minipage}}\end{center} \medskip

The ``ordinary'' {\sc Min-Covering} problem is $\min\{\tau(A):\nu(A)=0\}$;
it is a particular case of the {\GMC} problem
when we seek to minimize $M\nu(A)+\tau(A)$ for a large enough constant $M$.
Under certain assumptions, the {\sc Min-Covering} problem admits ratio $1+\ln \nu(\empt)$ \cite{Joh}.
Various generic covering problems are considered in the literature,
among them the {\sc Submodular Covering} problem \cite{Wol}, and several other types, cf. \cite{CKCV}. 
The variant we consider is inspired by the algorithms of 
Robins \& Zelikovsky \cite{RZ} for the {\sc Steiner Tree} problem,
and the analysis in \cite{GHNP} of this algorithm;
but, to the best of our knowledge, the explicit formulation of the {\GMC} problem given here is new.
Interestingly, our ratio for {\MPEC} is the same as that of \cite{RZ} 
for {\sc Steiner Tree} in quasi-bipartite graphs.

We call $\nu$ the {\bf potential} and $\tau$ the {\bf payment}.
The idea behind this interpretation and the subsequent greedy algorithm is as follows.
Given an optimization problem, the potential $\nu(A)$ is the value of some ``simple'' augmenting feasible solution for $A$. 
We start with an empty set solution, and iteratively try to decrease the potential by adding a set $B \subs U \sem A$ of minimum ``density'' --
the price paid for a unit of the potential. The algorithm terminates when the price $\geq 1$, since then we gain nothing from adding $B$ to $A$.
The ratio of such an algorithm is bounded by $1+\ln \f{\nu(\empt)}{\opt}$
(assuming that during each iteration a minimum density set  can be found in polynomial time).
So essentially the greedy algorithm converts ratio $\al=\f{\nu(\empt)}{\opt}$ into ratio $1+\ln \al$.
However, sometimes a tricky definition of the potential and the payment functions may lead to a smaller ratio.

Let ${\opt}$ be the optimal solution value of a problem instance at hand. Fix an optimal solution $A^*$.
Let $\nu^*=\nu(A^*)$, $\tau^*=\tau(A^*)$, so ${\opt}=\tau^*+\nu^*$. 
The quantity $\f{\tau(B)}{\nu(A)-\nu(A \cup B)}$ is called the {\bf density} of $B$ (w.r.t. $A$); this is the price 
paid by $B$ for a unit of potential.  
The {\sc Greedy Algorithm} (a.k.a. Relative Greedy Heuristic) for the problem starts with $A=\empt$ and while $\nu(A)>\nu^*$ 
repeatedly adds to $A$ a non-empty augmenting set $B$ that satisfies the 
following condition, while such $B$ exists:

\medskip \noindent
{\bf Density Condition:}  \ 
$\di \f{\tau(B)}{\nu(A)-\nu(A \cup B)} \leq \min\left\{1,\f{\tau^*}{\nu(A)-\nu^*}\right\}$.
\medskip

Note that since $\nu$ is decreasing
$\nu(A)-\nu(A \cup A^*) \geq \nu(A)-\nu(A^*)=\nu(A)-\nu^*$; hence if $\nu(A)>\nu^*$, then 
$\f{\tau(A^*)}{\nu(A)-\nu(A \cup A^*)} \leq \f{\tau^*}{\nu(A)-\nu^*}$
and there exists an augmenting set $B$ that satisfies the condition  
$\f{\tau(B)}{\nu(A)-\nu(A \cup B)} \leq \f{\tau^*}{\nu(A)-\nu^*}$, e.g., $B=A^*$.
Thus if $B$ is a minimum density set and $\f{\tau(B)}{\nu(A)-\nu(A \cup B)} \leq 1$,
then $B$ satisfies the Density Condition; otherwise, no such $B$ exists. 

\begin{theorem} \label{t:ga}
The {\sc Greedy Algorithm} achieves approximation ratio
$$
1+\f{\tau^*}{\opt} \ln \f{\nu_0-\nu^*}{\tau^*} =
1+\f{\tau^*}{\opt} \cdot \ln \left(1+\f{\nu_0-{\opt}}{\tau^*} \right) \ .
$$
\end{theorem}
\begin{proof}
Let $\ell$ be the number of iterations. 
Let $A_0=\empt$ and for $i=1, \ldots, \ell$ let 
$A_i$ be the intermediate solution at the end of iteration $i$ 
and $B_i=A_i \sem A_{i-1}$. Let $\nu_i=\nu(A_i)$, $i=0, \ldots,\ell$. Then:
$$
\f{\tau(B_i)}{\nu_{i-1}-\nu_i} \leq \min\left\{1,\f{\tau^*}{\nu_{i-1}-\nu^*}\right\}  \ \ \ \ \ i=1, \ldots, \ell
$$
Since $\nu$ is decreasing 
$$
\sum_{i=1}^\ell \tau(B_i)  \leq \sum_{i=1}^\ell \min\left\{1,\f{\tau^*}{\nu_{i-1}-\nu^*} \right\}(\nu_{i-1}-\nu_i)
$$
This is the lower Darboux sum of the function 
$\di
f(\nu)=\left\{ 
\begin{array}{ll} 
1                                        & \mbox{ if } \nu \leq \tau^*+\nu^* \\  
\f{\tau^*}{\nu-\nu^*} \ \ & \mbox{ if } \nu > \tau^*+\nu^*      
\end{array} 
\right .
$
in the interval $[\nu_\ell,\nu_0]$ w.r.t. the partition $\nu_\ell <\nu_{\ell-1} < \cdots < \nu_0$.
We claim that $\tau^*+\nu^* \geq \nu_\ell$.
For this, note that $\f{\tau(A^*)}{\nu(A)-\nu(A \cup A^*)} \geq 1$, thus since $\nu$ is decreasing
$\tau(A^*) \geq \nu(A)-\nu(A \cup A^*) \geq \nu(A)-\nu(A^*)$.
Consequently, $\sum_{i=1}^\ell \tau(B_i)$ is bounded by
$$
\int_{\nu_\ell}^{\nu_0} f(\nu) d\nu 
= \int_{\nu_\ell}^{\tau^*+\nu^*} 1d\nu +  \int_{\tau^*+\nu^*}^{\nu_0} \f{\tau^*}{\nu-\nu^*}d\nu  
= \tau^*+\nu^*-\nu_\ell +\tau^* \ln \f{\nu_0-\nu^*}{\tau^*}  
$$
Let $A=\bigcup_{i=1}^{\ell}B_i$ be the set computed by the algorithm. 
Since $\tau$ is sub-additive
$$
\tau(A) \leq \sum_{i=1}^\ell \tau(B_i) \leq \tau^*+\nu^*-\nu(A)  +\tau^* \ln \f{\nu_0-\nu^*}{\tau^*}  
$$
Thus the approximation ratio is bounded by
$\f{\tau(A)+\nu(A)}{\opt} \leq 1+\f{\tau^*}{\opt} \ln \f{\nu_0-\nu^*}{\tau^*}$.
\qed
\end{proof}

\section{Algorithm for general thresholds (Theorem~\ref{t:EC})} \label{s:EC}

Given an instance $G=(V,E),\t,R$ of {\AEC} 
the corresponding {\GMC} instance $U,\tau,\nu$ is defined as follows.
We put at each node $u \in V$ a large set of ``assignment units'', 
and let $U$ be the union of these sets of ``assignment units''. 
Note that to every $A \subs U$ naturally corresponds the assignment $\a$ where  
$a_u$ is the number of units in $A$ put at $u$. 
It would be more convenient to define $\nu$ and $\tau$ in terms of assignments, 
by considering instead of a set $A \subs U$ the corresponding assignment $\a$.

To define $\nu$ and $\tau$, let us recall the assignments $\q$ and $\c$ from the Introduction.
We have $c_v=q_v=0$ if $v \in V \sem R$ and for $u \in R$:
\begin{itemize}
\item[$\bullet$]
$\di q_u=\min_{e =uv \in E} t^e_u$ is the minimum threshold at $u$ of an edge in $E$ incident to $u$.
\item[$\bullet$]
$\di c_u=\min_{e =uv \in E}(t^e_u+t^e_v)-q_u$, so $c_u+q_u$ is the minimum value of an edge in $E$ incident to $u$.
\end{itemize}
We let $Q=\q(V)=\q(R)$ and $C=\c(R)$.
Note that $\c(R') \leq \th \q(R')$ for any $R' \subs R$; in particular, $C \leq \th Q$. 
For an assignment $\a$ that ``augments'' $\q$ 
let $R_{\q+\a}$  denote the set of terminals covered  by $E_{\q+\a}$. 
A natural definition of the potential and the payment functions would be 
$\tau(\a)=\a(V)$ and $\nu(\a)=(\c+\q)(R \sem R_{\q+\a})$ but this will enable to prove only ratio $1+\ln (\th+1)$.
We show a better ratio by adding to the potential in advance the ``fixed'' part $Q$.
We define
$$
\tau(\a)=\a(V) \ \ \ \ \ \ \nu(\a)=Q+\c(R \sem R_{\q+\a})  
$$
It is easy to see that $\nu$ is decreasing, $\tau$ is sub-additive, and $\tau({\bf 0})=0$.

The next lemma shows that the obtained {\GMC} instance is equivalent to the original {\AEC} instance.

\begin{lemma} \label{l:fe}
If $\q+\a$ is a feasible solution for {\AEC}  then $\tau(\a)+\nu(\a)=Q+\a(V)$.
If $\a$ is a feasible solution for {\GMC} then one can construct in polynomial time a feasible solution for {\AEC}  
of value at most $\tau(\a)+\nu(\a)$. In particular, both problems have the same optimal value,
and {\GMC} has an optimal solution $\a^*$ such that $\nu(\a^*)=Q$ and thus ${\opt}=\tau(\a^*)+Q$.
\end{lemma}
\begin{proof}
If $\q+\a$ is a feasible {\AEC} solution then $R_{\q+\a}=R$ and thus $\nu(\a)=Q$.
Consequently, $\tau(\a)+\nu(\a)=\a(V)+Q$.

Let now $\a$ be a {\GMC} solution.
The assignment $\q+\a$ has value $Q+\a(V)$ and activates the edge set $E_{\q+\a}$ that covers $R_{\q+\a}$.
To cover $R \sem R_{\q+\a}$, pick for every $u \in R \sem R_{\q+\a}$ an edge $uv$ with $t^{uv}_u+t^{uv}_v$ minimum.
Let $\b$ be an assignment defined by $b_u=c_u$ if $u \in R \sem R_{\q+\a}$ and $b_u=0$ otherwise.
The set of picked edges can be activated by an assignment $\q+\b$ that has value $Q+\c(R \sem R_{\q+\a})$.
The assignment $\q+\a+\b$ activates both edge sets and has value 
$Q+\a(V)+\c(R \sem R_{\q+\a})=\tau(\a)+\nu(\a)$, as required.
\qed
\end{proof}

For the obtained {\GMC} instance, let us fix an optimal solution $\a^*$ as in Lemma~\ref{l:fe}, 
so $\nu^*=Q$ and $\opt=\tau^*+Q$. Denote $\nu_0=\nu({\bf 0})=Q+\c(R)$, and note that $\c(R) \leq \th Q$.
To apply Theorem~\ref{t:ga} we need several bounds given in the next lemma.

\begin{lemma} \label{l:bound}
$\di \f{\opt}{\tau^*} \geq 1+\f{1}{\th}$, \ 
$\di \f{\nu_0}{\tau^*}\leq (\th+1)\left(\f{\opt}{\tau^*}-1\right)$, and \ 
$\di \f{\nu_0-\nu^*}{\tau^*} \leq \De+1$.
\end{lemma}
\begin{proof}
Note that 
$$
\tau^*+Q = {\opt} \leq \nu_0 \leq (\th+1) Q \ .
$$ 
In particular, $Q \geq \tau^*/\th$, and this implies the first bound of the lemma
$$
\f{\opt}{\tau^*}=1+\f{Q}{\tau^*} \geq 1+\f{1}{\th} \ .
$$
The second bound of the lemma holds since $\nu_0 \leq (\th+1)Q =(\th+1)(\opt-\tau^*)$.

The last bound of the lemma is equivalent to the bound $\c(R) \leq \tau^*(\De+1)$.
Let $J$ be an inclusion minimal edge cover of $R$ activated by $\q+\a^*$. 
Then $J$ is a collection ${\cal S}$ of node disjoint rooted stars with leaves in $R$.
Let $S \in {\cal S}$. By the definition of $\c$, $\di \a^*(S) \geq \max_{u \in R \cap S} c_u$,
thus  $\c(R \cap S) \leq |R \cap S| \a^*(S) \leq (\De+1)\a^*(S)$. Consequently,
$\di \c(R)=\sum_{S \in {\cal S}}\c(R \cap S) \leq (\De+1) \sum_{S \in {\cal S}}\a^*(S) \leq (\De+1) \a^*(V)$.
\qed
\end{proof}

We will show later that the {\sc Greedy Algorithm} can be implemented in polynomial time;
now we focus on showing that it achieves the approximation ratios stated in Theorem~\ref{t:EC}.
Substituting  Lemma~\ref{l:bound} second bound in Theorem~\ref{t:ga} second bound and denoting 
$x=\f{\opt}{\tau^*}-1$, we get that $x \geq 1/\th$ and that the ratio is bounded by 
$$
1+\f{\tau^*}{\opt} \cdot \ln \left(1+\f{\nu_0}{\tau^*}-\f{\opt}{\tau^*} \right) \leq 1+\f{\ln (\th x)}{x+1} =1+f(x)
$$
Consequently, the the ratio is bounded by $1+\max\{f(x): x \geq 1/\th\}$.
We now derive a formula for the maximum.
We have $\di \lim_{x \rightarrow \infty}f(x)=0$ (this can be shown using L'Hospital's Rule),
and $f(1/\th)=0$. Also:
$$
f'(x)=-\f{1}{{(x+1)}^2} \ln (\th x)+\f{1}{x+1} \f{1}{x}
$$
Hence $f'(x)=0$ if and only if $\f{1}{x+1}\ln (\th x)=\f{1}{x}$, namely, $x+1=x\ln (\th x)$.
For the analysis, we substitute $x \gets 1/x$, and get the equation $1+x=\ln (\th/x)$, where $0<x \leq \th$.
Since the function $x+1$ is strictly increasing and the function 
$\ln (\th/x)$ is strictly decreasing, this equation has at most one root; we claim that this root exists and
is in the interval $(0,\th]$.
To see this consider the function $h(x)=x+1-\ln (\th/x)$,
and note that $h$ is continuous and that $h(\th)=\th+1>0$ while $h(\eps)=\eps+1-\ln (\th/\eps)<0$
for $\eps>0$ small enough.

From this we get that the ratio is bounded by $1+\om(\th)$, where $\om(\th)$ 
is the root of the equation $x+1=\ln(\th/x)$.

Substituting  Lemma~\ref{l:bound} third bound in Theorem~\ref{t:ga} first bound
and observing that $\tau^* \leq {\opt}$ we get that the ratio is bounded by $1+\ln(\De+1)$.
In the case when $R$ is an independent set in $G$, it is easy to see that Lemma~\ref{l:bound} third bound improves to 
$\di \f{\nu_0-\nu^*}{\tau^*} \leq \De$, and we get ratio $1+\ln \De$ in this case.

\medskip

Finally, we show that the {\sc Greedy Algorithm} algorithm can be implemented in polynomial time.
As was mentioned in Section~\ref{s:GMC} before Theorem~\ref{t:ga}, 
we just need to perform in polynomial time the following two operations for any assignment $\a$:  
to check the condition $\nu(\a) > \nu^*$, and to find an augmenting assignment $\b$ of minimum density.

It is is easy to see that assignments $\q$ and $\c$ can be computed in polynomial time, 
and thus the potential $\nu(\a)=Q+\c(R \sem R_{\q+\a})$ can be computed in polynomial time, for any $\a$.
Let $\a^*$ be an optimal solution as in Lemma~\ref{l:fe}, and denote $\tau^*=\tau(\a^*)$ and $\nu^*=\nu(\a^*)=Q$.
Then the condition $\nu(A) > \nu ^*$ is equivalent to $\nu(\a) \geq Q$
and thus can be checked in polynomial time.

Now we show how to find an augmenting assignment $\b$ of minimum density.
Note that the density of an assignment $\b$ w.r.t. $\a$ is
$$
\f{\tau(\b)}{\nu(\a)-\nu(\a+\b)}=\f{\b(V)}{\c(R\sem R_{\q+\a})-\c(R \sem R_{\q+\a+\b})}=
\f{\b(V)}{\c(R_{\q+\a+\b} \sem R_{\q+\a})} \ .
$$

\begin{lemma} \label{l:R}
There exists a polynomial time algorithm that given an instance of 
{\AEC} and an assignment $\a$ finds an assignment $\b$ of minimum density. 
\end{lemma}
\begin{proof}
A {\bf star} is a rooted tree $S=(V_S,E_S)$ with at least one edge such that only its root $s$ may have degree $\geq 2$. 
We say that a star $S$ is a {\bf proper star} if all the leaves of $S$ are terminals.
We denote the terminals in $S$ by $R_S=R \cap V_S$.

Since $\q,\a$ are given assignments, we may simplify the notation by assuming that 
$R \gets R \sem R_{\q+\a}$ is our set of terminals, and that 
$\a \gets \q+\a$ is our given assignment.
Then the density of $\b$ is just $\f{\b(V)}{\c(R_{\a+\b})}$.
Let $\b^*$ be an assignment of minimum density, and let $J^* \subs E_{\a+\b^*}$ be an inclusion minimal $R_{\a+\b^*}$-cover.
Then $J^*$ decomposes into a collection ${\cal S}$ of node disjoint proper stars that collectively cover $R_{\a+\b^*}$. 
For $S \in {\cal S}$ let $\b^S$ be the optimal assignment such that $\a+\b^S$ activates $S$.
Since the stars in ${\cal S}$ are node disjoint
$$
\sum_{S \in {\cal S}} \b^S(V) \leq \b^*(V) \ \ \mbox{ and } \ \ \sum_{S \in {\cal S}}\c(R_S)=\c(R_{\a+\b^*}) \ .
$$
By an averaging argument, $\frac{\b^S(V)}{\c(R_S)} \leq \f{\b^*(V)}{\c(R_{\a+\b^*})}$ holds
for some $S \in {\cal S}$, and since $\b^*$ is a minimum density assignment, so is $\b^S$,
and $\frac{\b^S(V)}{\c(R_S)}=\f{\b^*(V)}{\c(R_{\a+\b^*})}$ holds.
Consequently, it is sufficient to show how to find in polynomial time 
an assignment $\b$ such that $\a+\b$ activates a proper star $S$ and $\frac{\b(V)}{\c(R_S)}$ is minimal.

We may assume that we know the root $v$ and the value $w=b_v$ of an optimal density pair $S,\b$;
there are at most $|V||E|$ choices and we can try all and return the best outcome.
Let $R_w=\{u \in R: \mbox{there is a } uv\mbox{-edge } e \mbox{ with } t^e_v \leq a_v+w\}$. 
For $u \in R_w$ let $b_u$ be the minimal non-negative number for which there is a $uv$-edge $e$ with 
$a_v+w \geq t^e_v$ and $a_u+b_u \geq t^e_u$. 
Then our problem is equivalent to finding $R_S \subs R_w$ with $\si(R_S)=\f{w+\b(R_S)}{\c(R_S)}$ minimum.
This problem can be solved in polynomial time, by starting with $R_S=\empt$ and while there is $u \in R_w \sem R_S$
with $\si(R_S+u)<\si(R_S)$, adding $u \in R_w \sem R_S$ to $R_S$ with $b_u/c_u$ minimum. 
\qed
\end{proof}

The proof of Theorem~\ref{t:EC} is complete. 

\section{Locally uniform thresholds (Theorem~\ref{t:lu})} \label{s:lu}

Here we consider the {\sc Bipartite} {\AEC} problem with locally uniform thresholds.
This means that each non-terminal $v \in V \sem R$ has weight $w_v$ 
and all edges incident to $v$ have the same threshold $t^v$; in the $\th$-bounded version $w_v \leq \th t^v$.
We consider a natural greedy algorithm that repeatedly picks a star $S$ that minimizes the average price paid for 
each terminal (the quotient of the optimal activation value of $S$ over $|R_S|$), and then removes $R_S$.
Each time we choose a star $S$ we distribute its activation value uniformly among its terminals, 
paying in the computed solution the average price for each terminal of $S$.

We now apply a standard ``set-cover'' analysis, cf. \cite{young}. 
In some optimal solution fix an inclusion maximal star $S^*$ with center $v$ and terminals 
$R_{S^*}$ covered by the algorithm in the order 
$r_k, r_{k-1}, . . . ,r_1$, where $r_k$ is covered first and $r_1$ last; we bound the algorithm payment for covering $R_{S^*}$.
Note that $1 \leq k \leq \De$.
Denote $w=w_v$ and let $t$ be the threshold of the terminals in $S^*$.
Let $S^*_i$ be the substar of $S^*$ with leaves $r_i, \ldots,r_1$.
At the start of the iteration in which the algorithm covers $r_i$, the terminals of $S^*_i$ are uncovered. 
Thus the algorithm pays for covering $r_i$ at most the average price paid by $S^*_i$, namely $(w+it)/i =w/i+t$. 
Over all iterations, the algorithm pays for covering $R_{S^*}$ at most $w H_k+kt$, while the optimum pays $w+kt$.
Thus the quotient between them is bounded by
$$
\f{w H_k +kt}{w+kt} =\f{w/t H_k +k}{w/t+k} \leq \f{\th H_k+k}{\th+k}=1+\f{\th(H_k-1)}{\th+k} \leq 
1+\max_{1 \leq k \leq \De} \f{H_k-1}{1+k/\th} \ .
$$
Since any optimal solution decomposes into node disjoint stars,
the last term bounds the approximation ratio, concluding the proof of Theorem~\ref{t:lu}.
We make some observations about this bound. Let $g(k)=\f{\th(H_k-1)}{\th+k}$. We have
$$
g(k+1)-g(k)=\f{\th}{(k+\th)(k+1+\th)}\left(2-H_k+\f{\th-1}{k+1} \right)
$$
Thus $g(k+1) \geq g(k)$ if and only if $2-H_k+\f{\th-1}{k+1} >0$. 
Hence if $k_\th$ is the smallest integer such that $H_k \geq 2+\f{\th-1}{k+1}$ then 
$\di \max_{1 \leq k \leq \De} g(k)=g(\min\{k_\th,\De\})$.
We do not have a more convenient formula of $\max_{k \geq 1} g(k)$ for arbitrary $\th$, 
but we can bound it using the inequality $H_k-1 \leq \ln k$. Then we have:
$$
\max_{1 \leq k \leq \De} \f{\th(H_k-1)}{k+\th} \leq \max_{x \geq 1} \f{\th\ln x}{x+\th}=\max_{x \geq 1} f(x) \ .
$$
Using fundamental calculus one can see that the maximum is attained when $x+\th=x \ln x$, 
and  substituting this in $f(x)$ we get
that $\max_{x \geq 1} f(x)=\f{\th}{\al}$ where $\al$ is the solution to the equation  $x+\th=x \ln x$. 
We have  $\al=\f{\th}{W(\th/e)}=\f{\th}{\om(\th)}$.
Thus the ratio is bounded by $1+\om(\th)$, a bound that we got before in Theorem~\ref{t:EC}.

If $\th=1$ then $k_\th=4$, since $H_3=11/6<2$ and $H_4=25/12>2$.
We have $g(4)=\f{13}{60}$, so for $\th=1$ we get ratio $1+g(4)=\f{73}{60}$. 
The ratio $\f{73}{60}$ is tight for unit thresholds, as shows the example in Fig.~\ref{f:tight}. 
The instance has $48$ terminals (in black), and two sets of covering nodes:
the upper $12$ nodes that form an optimal cover, and the bottom $13$ nodes. 
The bottom nodes have $3$ nodes of degree $4$, $4$ of degree $3$, and $6$ of degree $2$. 
The algorithm may start taking all bottom nodes, and only then add the upper ones, thus creating a solution of value $73$, 
instead of the optimum $60$.

\begin{figure} \centering
\includegraphics{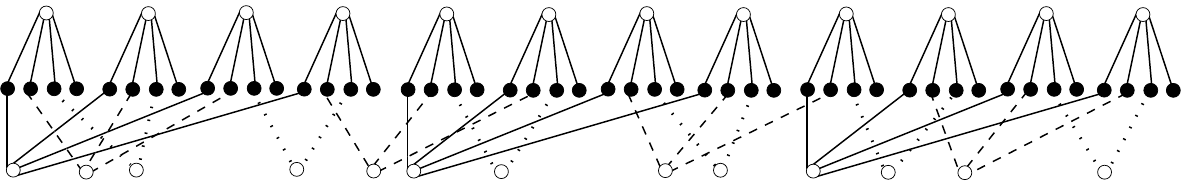}
\caption{Tight example of ratio $\f{73}{60}$ for unit thresholds.}
\label{f:tight}
\end{figure}

\section{Unit thresholds (Theorem~\ref{t:unit})} \label{s:unit}

Here we consider the case of unit thresholds when $t_u^e=t_v^{e}=1$ for every $uv$-edge $e$.
By a reduction from \cite{CKN}, we may assume that the instance is bipartite. 
Specifically, for any optimal assignment $\a$ we have $a_u=1$ for all $u \in R$, 
hence we can consider the residual instance obtained by removing the terminals covered by edges with both ends in $R$;
in the new obtained instance $R$ is an independent set, and recall that we may assume that $V \sem R$ is an independent set.

One can observe that in the obtained bipartite instance, 
$\a$ is an optimal solution if and only if $a_v \in \{0,1\}$ for all $v \in V$, $a_v=1$ for all $v \in R$,
and the set $C=\{v \in V \sem R:a_v=1\}$ covers $R$, meaning that $R$ is the set of neighbors of $C$.
Namely, our problem is equivalent to $\min\{|C|+|R|: C \subs V \sem R, C \mbox{ covers } R\}$.  
On the other hand the problem $\min\{|C|:C \subs V \sem R, C \mbox{ covers } R\}$ is essentially the (unweighted) 
{\SC} problem, and $C$ is a feasible solution to this {\SC} instance if and only if $C \cup R$ is the characteristic set 
of a feasible assignment for the {\AEC} instance. Note that both problems are equivalent w.r.t. their optimal solutions but
may differ w.r.t. approximation ratios, since if $C^*$ is an optimal solution to the {\SC} instance
then $\f{|C|+|R|}{|C^*|+|R|}$ may be much smaller than $\f{|C|}{|C^*|}$.

Recall that  a standard greedy algorithm for {\SC} repeatedly picks the center of a largest star and removes the star from the graph.  
This algorithm has ratio $H_k$ for $k$-{\SC}, where $k=\De$ is the maximum degree of a non-terminal (the maximum size of a set).
However, the same algorithm achieves a much smaller ratio $\f{73}{60}$ for {\AEC} with unit thresholds; 
the ratio $\f{73}{60}$ was established in \cite{CKN}, and it also follows from the case $\th=1$ in Theorem~\ref{t:lu}. 
In what follows we denote by $\al_k$ the best known ratio for $k$-{\SC}.
We have $\al_1=\al_2=1$ ($k=2$ is the {\sc Edge-Cover} problem) and $\al_3=4/3$ \cite{DF}.
The current best ratios for $k \geq 4$ are due to \cite{FY} (see also \cite{L,ACK}).
We summarize the current values of $\al_k$ for $k \leq 7$ in the following table.

\vspace*{-0.3cm}
\begin{table} [htbp] 
\begin{center}
\begin{tabular}{|c|c|c|c|c|c|c|c|c|c|c|}
\hline
$\al_1$        & $\al_2$        & $\al_3$              & $\al_4$                  & $\al_5$                   & $\al_6$                 & $\al_7$                
\\ \hline
 \ \ $1$ \ \   & \ \ $1$ \ \    & \ \ $\f{4}{3}$ \ \ & \ \ $\f{73}{48}$ \ \ & \ \ $\f{26}{15}$ \ \ & \ \ $\f{28}{15}$ \ \ & \ \ $\f{212}{105}$ \ \ 
\\ \hline
\end{tabular}
\end{center}
\caption{Current values of $\al_k$ for $k \leq 7$.}
\label{tbl:al}        \vspace*{-0.6cm}
\label{table:al}
\end{table}

We now show how these ratios for $k$-{\SC} can be used to approximate the {\AEC} problem with unit costs.
We start by describing a simple algorithm with ratio $1\f{67}{360}<\f{73}{60}$, that uses only the $k=2$ case.

\medskip
\begin{algorithm}[H] \caption{ratio $1\f{67}{360}$}
$A \leftarrow \emptyset$ \\
 \While{there exists a star with at least $3$ terminals}{
 add to $A$ and remove from $G$ the node-set of a maximum size star}
Add to $A$ an optimal solution of the residual instance
\end{algorithm} \medskip

We claim that the above algorithm achieves approximation ratio $1\f{67}{360}$ for {\AEC}
(a similar analysis implies ratio $H_k-\f{1}{6}$ for {\SC}).
In some optimal solution fix a star $S^*$ with terminals covered in the order $r_k, r_{k-1}, . . . ,r_1$,
where $r_k$ is covered first and $r_1$ last; we bound the algorithm payment to cover these terminals.
Let $S^*_i$ be the substar of $S^*$ with leaves $r_i, \ldots,r_1$.
At the start of the iteration when $r_i$ is covered, the terminals of $S^*_i$ are uncovered. 
Thus the algorithm pays for covering $r_i$ at most the density of $S^*_i$, namely, $(i+1)i=1+1/i$.
Over all iterations, the algorithm pays for covering $R_S$ at most $k+H_k$, while the optimum pays $k+1$.
If $k=1$ then the algorithm pays at most the amount of the optimum.
We claim that if $k \geq 2$ then in fact the payment is at most $k+H_k-1/6$.
If $k=2$ then the payment is at most $3 < 3+H_3-1/6$ (we pay $3$ if the star ``survives'' all the iterations before the last).
For $k \geq 3$, the pay for the last $3$ terminals is either:
$4/3$ for each of for $r_3,r_2$ and $2$ for $r_1$ (a total of $14/3$), or 
$4/3$ for $r_3$ and $3$ for $r_2,r_1$ (a total of $13/3$).
The maximum is $14/3=3+H_3-1/6$.  
Consequently, the ratio is bounded by 
$$
\max_{k \geq 2}\f{k+H_k-1/6}{k+1}=1+\max_{k \geq 2}\f{H_k-7/6}{k+1}=1+\max_{k \geq 2} g(k)
$$
By fundamental computations we have $g(k+1)-g(k) = \f{13/6-H_k}{(k+1)(k+2)}$.
Thus $g(k)$ is increasing iff $H_k <\f{13}{6}$. 
Since $H_4=\f{23}{12}<\f{13}{6}$ and $H_5=\f{137}{60}>\f{13}{6}$, we get that $\max_{k \geq 2} g(k)=g(5)=\f{67}{360}$, 
so we have ratio $1\f{67}{360}$. \\

We now show ratio $\f{8}{7}<\f{1555}{1347}<\f{7}{6}$. 
As in the greedy algorithm for {\SC}, we repeatedly remove an inclusion maximal set of disjoint stars with maximum number of leaves and pick the set of roots of these stars.
The difference is that each time stars with more than $k$ leaves are exhausted, we compute an $\al_k$-approximate solution $A_k$ for the remaining $k$-{\SC} instance;
we let $A_0=\empt$. 
This gives many {\SC} solutions, each is a union of the centers of stars picked and $A_k$; we choose the smallest one,
and together with $R$ this gives a feasible {\AEC} solution.
Formally, the algorithm is:

\medskip
\begin{algorithm}[H] \caption{ratio $\rho=\f{1555}{1347} < 1.1545$}
\For {$k\leftarrow \De$ downto $0$}
{
remove from $G$ a maximal collection of node disjoint $(k+1)$-stars \newline
\hspace*{-0.3cm} let $C_{k+1}$ be the set of the roots of the stars removed so far  \\ 
compute an $\al_k$-approximate $k$-{\SC} solution $A_k$ in $G$
} 
Return the smallest set $C_{k+1} \cup A_k$, $k \in \{\De,\ldots,0\}$
\end{algorithm} \medskip

Since we claim ratio $\f{1555}{1347} < \f{8}{7}$, at iterations when $k \geq 7$ step~3 can be skipped,
since then we can apply a standard ``local ratio'' analysis \cite{BY}.
Indeed, when a star with $k \geq 7$ terminals is removed, the partial solution value increases by $k+1$
while the optimum decreases by at least $k$. Hence for $k \geq 7$ it is a $\f{k+1}{k} \leq \f{8}{7}$ local ratio step.
Consequently, we may assume that $\De \leq 6$, provided that we do not claim ratio better than $8/7$.

Let $r=|R|$. Let $\tau$ be the optimal value to the initial {\SC} instance.
At iteration $k$ the algorithm computes a solution of value at most $\al_k \tau+r+|C_{k+1}|$.
Thus we get ratio $\rho$ if $\rho(r+\tau) \geq \al_k \tau+r+|C_{k+1}|$ holds for some $k \leq 6$. 
Otherwise,
\begin{eqnarray*}
\rho(r+\tau)   &< & \al_6\tau+r \\
\rho(r+\tau)   &< & \al_5\tau+r+|C_6| \\
\rho(r+\tau)   &< & \al_4\tau+r+|C_5| \\
\rho(r+\tau)   &< & \al_3\tau+r+|C_4| \\
\rho(r+\tau)   &< & \al_2\tau+r+|C_3| \\
\rho(r+\tau)   &< & \al_1\tau+r+|C_2| \\
\rho(r+\tau)   &< & \hspace{1.0cm} r+|C_1| 
\end{eqnarray*}
Denote $\si=\al_1+\cdots +\al_5=\f{1581}{240}$. 
Note that $|C_1|+ \cdots + |C_6| = r$, since in this sum the number of stars with $k$ leaves is summed exactly $k$ times, $k=1, \ldots,6$. 
The first inequality, and the inequality obtained as the sum of the other six inequalities gives the following two inequalities:
\begin{eqnarray*}
\rho(r+\tau)    & < & \al_6 \tau+r \\
6 \rho(r+\tau) & < & \si \tau+7r 
\end{eqnarray*}
Dividing both inequalities by $\tau$ and denoting $x=r/\tau$ gives:
\begin{eqnarray*}
\rho(x+1)    & < & \al_6+x \\
6 \rho(x+1) & < & \si +7x 
\end{eqnarray*}
Since $\rho>1$ and $7>6\rho$ this is equivalent to:
$$
\f{6 \rho-\si}{7-6\rho} <x < \f{\al_6-\rho}{\rho-1}
$$
We obtain a contradiction if $\rho$ is the solution of the equation 
$\f{6 \rho-\si}{7-6\rho}=\f{\al_6-\rho}{\rho-1}$, namely
$$
\rho=\f{7\al_6-\si}{6\al_6-\si+1}=1+\f{\al_6-1}{6\al_6-\si+1}=1+\f{208}{1347}=\f{1555}{1347} \ .
$$

This concludes the proof of Theorem~\ref{t:unit}.


\end{document}